\documentclass[cmp]{cls}
\usepackage{latexsym}
\usepackage{graphics}
\usepackage{amsfonts}

\newcommand{\Torus}{\mathbb{T}}
\newcommand{\Real}{\mathbb{R}}
\newcommand{\Amoeba}[1]{\mathcal{A}_{#1}}
\newcommand{\Newton}[1]{\mathcal{N}_{#1}}
\newcommand{\Contour}[1]{\mathcal{C}_{#1}}

\spnewtheorem*{theorem*}{Theorem}{\bf}{\it}
\spnewtheorem*{remark*}{Remark}{\bf}{\it}
\spnewtheorem{state}{Statement}{\bf}{\it}
\newcommand{\Log}{\textup{Log}}

\begin{document}
\title{Amoebas of complex hypersurfaces in statistical thermodynamics}
\author{\fbox{Mikael Passare\inst{1}} \and Dmitry Pochekutov\inst{2} \and August Tsikh\inst{3}
}                     
\institute{
Department of Mathematics, Stockholm University, Stockholm, Sweden.\and
Institute of Core Undergraduate Programmes, Siberian Federal University, Russia.\\
E-mail: potchekutov@gmail.com\and
Institute of Mathematics, Siberian Federal University, Krasnoyarsk, Russia.\\
E-mail: tsikh@lan.krasu.ru  
}
%
%
%
\maketitle
\begin{abstract}
The amoeba of a complex hypersurface is its image under a logarith\-mic projection. 
A number of properties of algebraic hypersurface amoebas are carried over to the 
case of  transcendental hypersurfaces. We demonstrate the potential that amoebas
can bring into statistical physics by considering the problem of energy distribution in a quantum
thermodynamic ensemble. The spectrum $\{\varepsilon_k\}\subset \mathbb{Z}^n$
of the ensemble is assumed to be multidimensional; this leads us to the notions of a multidimensional
temperature and a vector of differential thermodynamic forms.  Strictly speaking, in the paper
we develop the multidimensional Darwin and Fowler method and give the description
of the domain of admissible average values of energy for which the thermodynamic limit exists.

\end{abstract}
\section{Introduction}

The amoeba of a complex hypersurface~$V$ defined in a Reinhardt domain is the image of $V$
in the logarithmic scale. The notion of the amoeba of an algebraic hypersurface, introduced
in~\cite{Gelfand}, plays the fundamental role in the study of zero distributions of
polynomials in $\mathbb{C}^n$. Over the last decade, the amoebas proved to be a useful tool
and a convenient language in the diverse questions; such as the  classification of topological types
of Harnack curves~\cite{Mikhalkin}, description of phase diagrams of dimer models~\cite{KO,KOS},
study of the asymptotic behavior of solutions to multidimensional difference equations~\cite{LPT}.
Adelic (non-archimedeam) amoebas turned out to be helpful in the computation of nonexpansive
sets for  dynamical systems~\cite{EKL}.

The main purpose of the present paper  is to demonstrate the advantages of using the amoebas
in statistical physics. As an example of such usage  we consider the statistical problem of finding
the preferred states  of the thermodynamic ensemble when its spectrum is discrete. 

In the classical formulation of this problem, which was studied by Maxwell, Boltzmann and Gibbs, it is assumed
that the energy levels occupied by the ensemble systems form a discrete \textit{one}-\textit{dimensional}
spectrum  $\{\varepsilon_k\}\subset \mathbb{N}=\{0,1,2,\ldots\}$  (see, for example, \cite{Sh,Zo}). By contrast,
we consider the case of the  \textit{multidimensional} spectrum $\{ \varepsilon_k\}\subset \mathbb{Z}^n,\ n> 1$, and then
such important notions as the temperature and the differential thermodynamic form become vector quantities. 

In fact,  the major part of the paper is devoted to the generalization  of the asymptotic Darwin-Fowler method~\cite{DF1,DF2},
that gives a way to describe the state of a quantum thermodynamic ensemble with the multidimensional spectrum.
For this purpose, we introduce the notion of an amoeba of a general (not only algebraic) complex hypersurface
and describe the structure of the amoeba complement (Theorem~\ref{thm:1}).  Next, we prove an asymptotic formula (Theorem~\ref{thm:2}) for the diagonal Laurent coefficient of a meromorphic function;  the polar hypersurface, its amoeba and the
logarithmic Gauss mapping are significantly used in the proof.

There are two main reasons motivating to apply methods of the theory of amoebas to the asymptotic investigation of the Laurent
coefficient of a meromorphic function in several complex variables.  First, the connected components of the amoeba
complement are in one-to-one correspondence with the Laurent expansions of a meromorphic function centered at the origin;
and, moreover, define their domain of convergence. Second, by the multidimensional residues, the asymptotics of the
Laurent coefficient is represented by the oscillating integral over a chain on a polar hypersurface~$V$. 
In the logarithmic scale, the critical points of the phase function of such integral comprise  the contour of the amoeba of $V$.

Thus, our generalization of the Darwin-Fowler method (Sect.~\ref{sec:6}) is grounded on  Theorems~\ref{thm:1} and \ref{thm:2}.
Theorem~\ref{thm:3} provides the asymptotics of the average values for occupation  numbers of  energy $\varepsilon_k$ from a given spectrum. These average values are  expressed by the  Laurent coefficients of the meromorphic function constructed
by means of the partition  function of an ensemble. Although Theorem~\ref{thm:3} requires tricky 
integration techniques, its statement  is a quite expected generalization of the Darwin-Fowler results. 
This is not the case with Theorem~\ref{thm:4},  which is totally inspired by the geometry brought in our  investigation
 by the theory of amoebas. Theorem~\ref{thm:4}  gives the answer to the question  whether an average energy of an  ensemble
permits the thermodynamical limit. Namely, the domain of admissible average  energies coincides with the  interior of the convex hull
of the spectrum.

\section{Amoebas of complex hypersurfaces}
\label{sec:1}

For convenience we shall denote by
$\Torus^n$ the set $\left(\mathbb{C}\setminus\{0\}\right)^n$.

\begin{definition}[\cite{Gelfand}]
The \emph{amoeba} $\mathcal{A}_V$ of a complex algebraic hypersurface
$$V=\{ z\in\Torus^n: Q(z)=0 \}$$
(or of the polynomial $Q$) is the image of $V$ under the mapping $\Log\colon\Torus^n\to\Real^n$,
determined by the formula
$$\Log\colon (z_1,\dots,z_n)\mapsto(\log|z_1|,\dots,\log|z_n|).$$
\end{definition}

The term amoeba is motivated by the specific appearance of $\Amoeba{V}$ in the case $n=2$. It has a shape with thin tentacles going off to infinity
(see Fig.~\ref{fig:1}). The complement $\mathbb{R}^n\setminus \mathcal{A}_V$ consists of a finite number of connected components, which are open and convex \cite{Gelfand}. The basic results on amoebas of algebraic hypersurfaces can be found in
\cite{Mikhalkin,FPT,PaRu,PT}.

\begin{figure}
\centering
\resizebox{0.75\textwidth}{!}{%
  \includegraphics{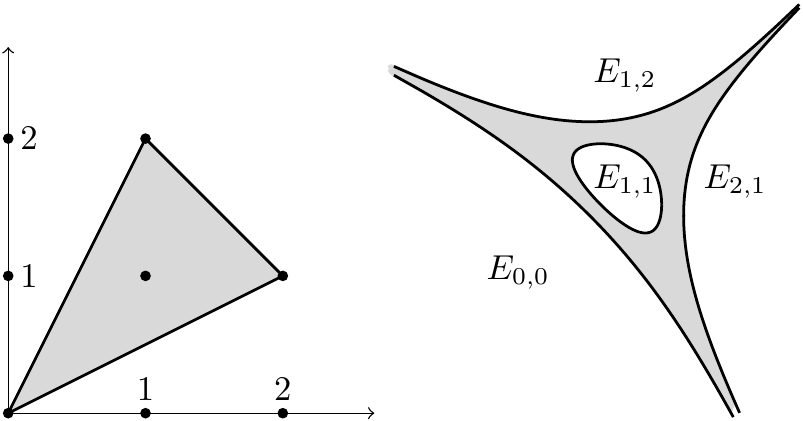}
}
\caption{The Newton polytope (left) and the amoeba with its complement components $E_\nu$ (right) for the polynomial
$Q(z)=z_1^2 z_2 - 4z_1z_2 + z_1z_2^2 + 1$.}
\label{fig:1}       
\end{figure}

We denote by $\Newton{Q}$ the \emph{Newton polytope} of the polynomial $Q$, that is, the convex hull in $\mathbb{R}^n$ of all the
exponents of the monomials occurring in the polynomial $Q$. For each integer point $\nu\in \mathcal{N}_Q$ we define the dual cone
$C_\nu$ to the polytope $\mathcal{N}_Q$ at the point $\nu$ to be the set
$$
C_\nu=\{s\in \mathbb{R}^n: \left<s,\nu\right>=\max_{\alpha\in N_Q} \left<s,\alpha\right> \}.
$$ 
We recall that the \emph{recession cone} of a convex set $E\subset \mathbb{R}^n$ is the largest cone, which after a suitable translation
is contained in $E$. The connection between the combinatorics of the Newton polytope $\mathcal{N}_Q$ of the polynomial $Q$ and the 
structure of the complement of the amoeba $\mathcal{A}_V$ is described by the following result.

\begin{theorem*}[\cite{FPT}]
On the set $\{ E\}$ of connected components of the complement $\mathbb{R}^n\setminus \mathcal{A}_V$ there exists an injective order
function
$$\nu\colon \{E\} \to \mathbb{Z}^n\cap \Newton{Q}$$
such that the dual cone $C_{\nu(E)}$ to the Newton polytope at the point $\nu(E)$ is equal to the recession cone of the component $E$.
\end{theorem*}

This means that the connected components of the complement $\mathbb{R}^n\setminus \mathcal{A}_V$ can be labelled as $E_\nu$ by
means of the integer vectors $\nu=\nu(E) \in \Newton{Q}$ (see Fig.~\ref{fig:1}).

The value $\nu(E)$ of the order function allows for two interpretations. On the one hand, $\nu(E)$ is the gradient
of the restriction to $E$ of the Ronkin function for the polynomial $Q$ (see \cite{PaRu}). The Ronkin function is a multidimensional analogue of Jensen's function and finds numerous applications in the theory of value distribution of meromorphic functions.
On the other hand,  components of the vector $\nu(E)$  are the linking numbers
of the basis loops in the torus $\textup{Log}^{-1}(x)$, for any $x\in E$, and the hypersurface $V$ (see \cite{FPT} or \cite{Mikhalkin}).

\begin{remark*}
The set $\textup{vert}\, \mathcal{N}_Q$ of vertices of the polytope $\Newton{Q}$ belongs to the image of the order function $\nu$. In other words, 
for each vertex $\beta\in\Newton{Q}$ there is a component $E_\beta$ with recession cone $C_\beta$ (\cite{Gelfand,MY}). The
existence of components $E_\nu$ corresponding to other integer points $\nu\in \Newton{Q}\setminus \textup{vert}\, \mathcal{N}_Q$ depends on the coefficients of the polynomial $Q$.
\end{remark*}  

There is a bijective correspondence between the connected components $\{E_\nu\}$ of the complement $\mathbb{R}^n\setminus \mathcal{A}_V$ 
and the Laurent expansions (centered at the origin) of an irreducible rational fraction $F(z)=P(z)/Q(z)$ 
(see \cite[Sect.~6.1]{Gelfand}). The sets $\Log^{-1}(E_{\nu})$ are the domains of convergence for the corresponding Laurent expansions. One may therefore label such an expansion using the components of the amoeba complement, or using the integer points in
the Newton polytope. For instance, the Taylor expansion of a function that is holomorphic at the origin will always correspond 
to the vertex of the Newton polytope $\Newton{Q}$ with coordinates  $(0,\ldots, 0)$.

In Sects.~\ref{sec:4}-\ref{sec:6} we shall see that, when working with partition functions, one needs to consider amoebas also of non-algebraic complex hypersurfaces.
Let $Q$ be a Laurent series in the variables $z=(z_1,\ldots, z_n)$:
$$
Q(z)=\sum_{\alpha\in A\subset \mathbb{Z}^n} a_\alpha z^{\alpha}\,.
$$
We assume that its domain of convergence is non-empty, and that $Q(z)\not\equiv 0$. 
We shall also make the assumption that $Q$ does have zeros in $G\cap \Torus^n$. Let
$$
V=\{z\in G\cap \Torus^n: Q(z)=0\}
$$
be the hypersurface given by the zeros of the analytic function $Q(z)$. The amoeba for $V$ is defined as in the algebraic case: 
$\mathcal{A}_V=\Log(V)$.

We introduce the notation $\mathcal{G}=\Log (G)$ for the image of the convergence domain $G$ of the series $Q$. It is well known that
$\mathcal{G}$ is a convex domain. In the algebraic case, when $Q$ is a polynomial, the set $\mathcal{G}$ is all of $\mathbb{R}^n$, and
the amoeba $\mathcal{A}_V$ is a proper subset of $\mathcal{G}$. In the general case it may well happen that there is an equality
$\mathcal{A}_V=\mathcal{G}$. To avoid this situation, we require that the summation support $A$ of the series $Q$ lies in some acute
cone, that is, the closure $\mathcal{N}$ of the convex hull $\textup{ch}(A)$ does not contain any lines.

\begin{theorem}
\label{thm:1}
If for the series $Q$ the set $\mathcal{N}=\overline{\textup{ch}(A)}$ does not contain any lines, then the complement
$\mathcal{G}\setminus \mathcal{A}_V$ is non-empty. To the set $\{\nu\}$ of vertices of the polyhedron $\mathcal{N}$ there corresponds
a family $\{E_\nu\}$ of pairwise distinct connected components of the complement $\mathcal{G}\setminus \mathcal{A}_V$. The dual cone
$C_\nu$ to $\mathcal{N}$ at the vertex $\nu$ coincides with the recession cone for $E_\nu$.
\end{theorem} 

\begin{proof}
Assumption of the theorem implies that the set of the vertices $\textup{vert}(\mathcal{N})$ is nonempty.
The argument is similar to the one for the algebraic case (when $Q$ is a polynomial and $\mathcal{G}=\mathbb{R}^n$) that is
given in \cite{FPT} and \cite{MY}. First one shows that for each vertex  $\nu\in\mathcal{N}$ a suitable translate of the cone $C_\nu$ is disjoint from  $\mathcal{A}_V$, 
so that one can associate with the vertex $\nu$ the component $E_\nu$ of the complement $\mathcal{G}\setminus \mathcal{A}_V$, which contains this translated cone. Here the only difference is that, when $\mathcal{G}\neq \mathbb{R}^n$, one must show that the translated cones are contained
in $\mathcal{G}$. This follows from the fact that the dual cones $C_\nu$ at the vertices of $\mathcal{N}$ all lie in the cone $-C^{\vee}(\mathcal{N})$, 
where $C^{\vee}(\mathcal{N})$ is the dual cone of the recession cone $C(\mathcal{N})$ of $\mathcal{N}$, together with the multidimensional Abel 
lemma \cite{PST}, which says that the cone $-C^{\vee}(\mathcal{N})$ lies in the recession cone of the domain $\mathcal{G}$.

Next, just as in \cite{MY}, one associates to the collection of $n$-cycles $\Gamma_\nu=\Log^{-1} (x_\nu)$, with the point $x_\nu$ taken in the
translate of $C_\nu$, a collection of de Rham dual $n$-forms $\omega^\mu$ which are meromorphic in $G\cap \mathbb{T}^n$ with poles on $V$. Namely, we choose
$$
	\omega^\mu=\frac{1}{(2\pi\imath)^n}\cdot\frac{a_\mu z^\mu}{Q(z)}\cdot \frac{dz_1}{z_1}\wedge\ldots \wedge \frac{d z_n}{z_n}, 
	\ \mu\in \textup{vert}(\mathcal{N})
$$
(recall, that $a_\mu$ is the Laurent coefficient of $Q$).  For points $z\in \Gamma_\nu$ we have
$|a_\nu z^\nu|>|g_\nu(z)|$, where $g_\nu(z)=Q(z)-a_\nu z^\nu$. Hence, the meromorphic function  $1/Q(z)$ can be
developed into a geometric progression 
$$
	\frac{1}{Q(z)}=\sum_{k=0}^\infty (-1)^k \frac{g^k_{\nu}(z)}{(a_\nu z^\nu)^{k+1}},
$$
uniformly converging on $\Gamma_\nu$, and one has
$$
	\int_{\Gamma_\nu} \omega^{\mu}=\sum_{k=0}^\infty \frac{(-1)^k}{(2\pi\imath)^n}\int_{\Gamma_\nu}\frac{a_\mu z^\mu}{a_\nu z^\nu}\cdot
	\left(\frac{g_\nu(z)}{a_\nu z^\nu}\right)^k \cdot\frac{dz_1}{z_1}\wedge\ldots \wedge \frac{d z_n}{z_n}.
$$
The leading term of $Q(z)$ with respect to the orders, defined by weight vectors from $C_\nu$,
is equal to $a_\nu z^\nu$.  This yields that all the integrals in the sum vanish for  $\nu\neq \mu$;
and if $\nu=\mu$, the  only one nonzero summand occurs for $k=0$ and equals $1.$
Therefore,
$$
	\int_{\Gamma_\nu} \omega^{\mu}=\delta_{\nu\mu},
$$
and by the de Rham duality~\cite{Le}  the cycles $\Gamma_\nu$, $\nu \in \textup{vert}\, \mathcal{N}$
are  linearly independent in the homology group $H_n((G\cap \mathbb{T}^n) \setminus V)$.
The cycles $\textup{Log}^{-1}(x)$ for $x$  from the same connected component of $\mathcal{G}\setminus \mathcal{A}_V$
are homologically equivalent, this implies that the connected components $\{E_\nu \}_{\nu\in \textup{vert}(\mathcal{N})}$
are pairwise distinct. Since the $n$-dimensional cones of a fan dual to $C(\mathcal{N})$
coincide with the cones $C_\nu$ and $C_\nu \subset E_\nu$, one has that $C_\nu$  coincides  with 
the recession cone for $E_\nu$.\qed
\end{proof}

\section{ The amoeba contour and the logarithmic Gauss mapping}
\label{sec:2}

In Sect.~\ref{sec:1} we saw that certain information about the position of the amoeba of a complex hypersurface is given by the combinatorics of the integer points of the Newton polytope (or polyhedron) of the polynomial (or series) that defines this hypersurface. Here we shall describe an object associated with the amoeba, that reflects the differential geometry of the hypersurface. The study of this object can be carried out with more analytic methods.

The \emph{contour} $\Contour{V}$ of the amoeba $\Amoeba{V}$ is defined (see \cite{PT}) as the set of critical values of the mapping
$\Log\colon V\to \Real^n$, that is, the mapping $\Log$ restricted to the hypersurface $V$. We observe that the boundary $\partial \Amoeba{V}$ is included in the contour $\Contour{V}$, but the inverse inclusion does not hold in general.
Note, a contour of an amoeba for Harnack's curve coincides with a boundary of the amoeba~\cite{Mikhalkin,LPT} (there
is the amoeba of the Harnack curve on Fig.~\ref{fig:1}). Herewith, a real section $V\cap \mathbb{R}^2$ of  Harnack's curve
consists of fold critical points of the projection $\textup{Log}: V\mapsto \mathcal{A}_V$. 
Fig.~\ref{fig:2} depicts the amoebas of the complex curves, which contours do not coincide with the boundaries, besides that
the points $a,b,c$ and $d$ are the images of Whitney pleats.

\begin{figure}
\centering
\resizebox{0.75\textwidth}{!}{%
  \includegraphics{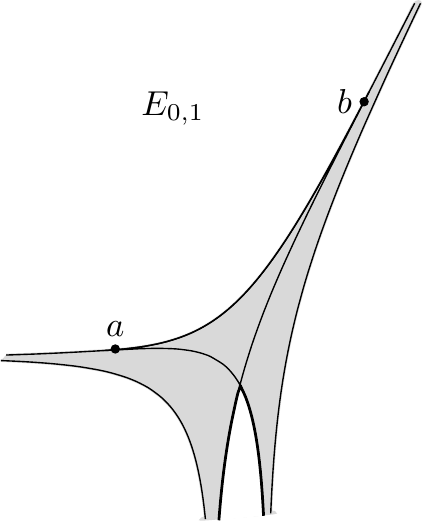}\hskip 1cm
  \includegraphics{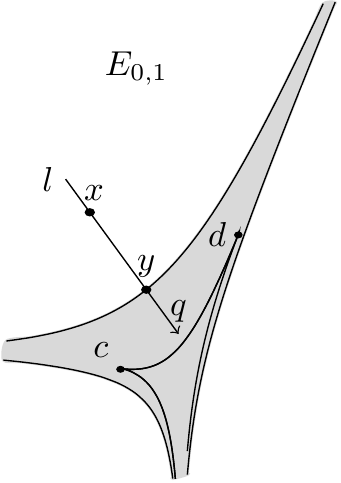}
}
\caption{The amoebas and their contours  for the graphs of polynomials
$1-2z-3z^2$ (left) and $1+z+z^2+z^3$ (right, 
the normal line $l$ to $\partial E_{0,1}$ with a directional vector $q$
and points $x$, $y$ illustrate the proof of the Theorem~2).}
\label{fig:2}       
\end{figure}

We recall (see \cite{Mikhalkin,Kapranov}) that the \emph{logarithmic Gauss mapping} of a complex hypersurface $V\subset \mathbb{T}^n$ is defined to be the mapping
$$
\gamma=\gamma_V\colon \textup{reg}\,V\to \mathbb{CP}_{n-1}\,,
$$
which to each regular point $z\in \textup{reg}\, V$ associates the complex normal direction to the image
$\log(V)$ at the point $\log(z)$. (Here $\log$, in contrast to $\Log$, denotes the full complex coordinatewise logarithm.) 
The image $\gamma(z)$ does not depend on the choice of branch of $\log$ and it is given in coordinates by the 
explicit formula \cite{Mikhalkin}:
$$
\gamma(z)= \left(z_1 Q'_{z_1}(z):\ldots : z_n Q'_{z_n}(z)\right)\,.
$$

The connection between the contour $\Contour{V}$ and the logarithmic Gauss mapping is given as follows.
\begin{proposition}[\cite{Mikhalkin}]
\label{prop:1}
The contour $\Contour{V}$ is expressed by the identity
$$\Contour{V}=\Log\left(\gamma^{-1}(\mathbb{RP}_{n-1})\right)\,.$$
In other words, the mapping $\gamma$ sends the critical points $z$ of $\left.\Log\right|_V$ to real
direction $\gamma(z)$ which is orthogonal to the contour $\Contour{V}$ at $\Log\, z$.
\end{proposition}
The inverse
$z=\gamma^{-1}(q)$ of the logarithmic Gauss mapping is given by the solutions to the system of equations
\begin{equation}\label{eq:IG} 
\left\{
\begin{array}{l}
Q(z)=0\,,\\
q_n z_j Q'_{z_j}-q_j z_n Q'_{z_n}=0\,,\quad\ j=1,\ldots, n-1\,.\\
\end{array}
\right.
\end{equation}

For a fixed vector $q\in \mathbb{Z}^n_{*}=\mathbb{Z}^n\setminus\{0\}$ the solutions to the system (\ref{eq:IG}) consist of
the points $z\in V$ at which the Jacobian of the mapping $(Q(z),z^q)$ has rank $\leq 1$, which means that the following
statement holds.

\begin{proposition}
\label{prop:ku}
A point $w\in\textup{reg}\,V$ is a critical point for the monomial function $\left.z^q\right|_V$ if and only if the logarithmic Gauss
mapping takes the value $q$ at $w$, that is,  $\gamma(w)=q$.
\end{proposition}
	
Notice that if $V$ is the graph of a function of $n$ variables $z=(z_1,\ldots, z_n)$, so that it is the zero set of the function $Q(z,w)=w-f(z)$, 
then the logarithmic Gauss mapping is given in the affine coordinates $s_j=q_j/q_{n+1},\ j=1,\ldots, n$ of
$\mathbb{CP}_n$ by the formula
\begin{equation}
\label{eq:n2}
z_j \frac{f'_{z_j}}{f}=-s_j\,,\quad j=1,\ldots, n\,.
\end{equation}
\section{Asymptotics of Laurent coefficients}
\label{sec:3}

Let $E$ be a connected component of the amoeba complement with smooth boundary $\partial E$. The cone generated by the outward normals to 
$\partial E$ will be called the \emph{component cone} of $E$ and denoted by $K_E$. It is clear that $K_E$ is a cone over the image of $\partial E$
under the ordinary Gauss mapping $\sigma:\partial E\to S^{n-1}$.

\begin{definition}
\label{def:simple}
The smooth boundary $\partial E$ of a connected component $E$ is said to be \emph{simple} if for each $x\in\partial E$ the real torus $\Log^{-1}(x)$
intersects $V$ in a unique point, and if moreover the logarithmic Gauss mapping $\gamma$ of the hypersurface $V$ is locally invertible at this
intersection point.
\end{definition}

The following result is a consequence of Lemmas 1 and 2 in the paper \cite{PoTs}, which exhibits a class of simple boundaries in the case where
$V=\Gamma_f$ is the graph over the convergence domain of a power series $f(z)=\sum\limits_{\alpha\in A\subset \mathbb{N}^n} \omega_\alpha z^\alpha$.

\begin{proposition}
\label{prop:3}
If $\bar{0}\in A$, the coefficients $\omega_\alpha$ are positive, and the set $A$ generates the lattice $\mathbb{Z}^n$ as a group,
then the boundary of the component $E_{\bar{0},1}$ of the complement of the amoeba $A_{\Gamma_f}$  is simple.
\end{proposition}

As it follows from an example of a polynomial $f=1-2z_1-3z_1^2$ (see Fig.~\ref{fig:2}) the condition of
coefficients $\omega_\alpha$ to be positive is essential in Proposition~\ref{prop:3}. Namely, the preimage
of the inner point of the arc $(a,b)\subset \partial E_{0,1}$ consists of two points on the graph $\Gamma_f$,
and the boundary point $a$ or $b$ have one preimage on $\Gamma_f$, but the logarithmic Gauss mapping has no inverse at $a$
and $b$.

 In view of the convexity and smoothness of $\partial E$ each point $x\in\partial E$ 
is the preimage $x=\sigma^{-1}(q)$ of a point $q\in K_E$.

We consider the expansion of the meromorphic function $F=P(z)/Q(z)$ in a Laurent series
\begin{equation}
\label{eq:L1}
F(z)=\sum_{\alpha\in \mathbb{Z}^n} c_\alpha z^\alpha\,,
\end{equation}
that converges in the preimage $\textup{Log}^{-1}(E)$ of a complement component $E$ of the amoeba of the polar hypersurface
$V=\{z:Q(z)=0\}$ of $F$. For a fixed $q\in \mathbb{Z}^n_{*}$ we define the diagonal sequence $c_{q\cdot k}=c_{(q_1,\ldots, q_n)\cdot k}$ of Laurent coefficients $c_\alpha$ from (\ref{eq:L1}).

\begin{theorem}
\label{thm:2}
Let the boundary $\partial E$ be simple. Then for each $q\in \mathbb{Z}_*^n\cap K_E$ the diagonal sequence  $\{c_{q\cdot k}\}$ has the
asymptotics
\begin{equation}
\label{eq:asymp}
c_{q\cdot k}= k^{\frac{1-n}{2}} \cdot z^{-q\cdot k}(q)\cdot \left\{C(q)+O(k^{-1})\right\}\,,
\end{equation}
as $k\to +\infty$. Here $z(q)=V\cap\textup{Log}^{-1}(\sigma^{-1}(q))$, and the constant $C(q)$ vanishes only when $P(z(q))=0$.
\end{theorem}

\begin{proof}
The idea of the proof is to choose the cycle of integration $\Log^{-1} (x)$ in the Cauchy formula
\begin{equation}\label{eq:n3}
c_{q\cdot k}=\frac{1}{(2\pi \imath)^n}\int_{\Log^{-1} (x)} \frac{F(z)}{z^{q\cdot k}}
\frac{dz_1}{z_1}\wedge\ldots\wedge\frac{dz_n}{z_n}\,,\quad x\in E\,,
\end{equation}
for those $x$ that lie near the point  $y=\Log z(q)\in \partial E$ on the line $l=\{y+qt: t\in\mathbb{R}\}$, which is tranversal to $\partial E$ (see Fig.~\ref{fig:2}). 
In view of the assumed simplicity of $\partial E$, the torus  $\Log^{-1}(y)\subset \Log^{-1}(l)$ intersects $V$ in a unique point, and $\Log^{-1}(l)$
intersects $V$ in a neighborhood of $z(q)$ along a $(n-1)$-dimensional chain $h\subset V$. By means of residue theory one shows (see \cite{Tsikh2}
for the case $n=2$) that, as a function of the parameter $k$, the integral (\ref{eq:n3}) is asymptotically equivalent, as $k\to+\infty$, to the oscillatory
integral
$$
2\pi \imath \int_h \textup{res}\,\omega \cdot e^{-k\left<q,\log z\right>},
$$
where
$$
\omega=\frac{1}{(2\pi \imath)^n} \frac{P(z)}{Q(z)} \frac{dz_1}{z_1}\wedge\ldots\wedge\frac{dz_n}{z_n}\,,
$$
and $\textup{res}\,\omega =Q\omega/dQ$ denotes the residue form for $\omega$. The phase $\varphi=\left<q, \log z\right>=\log z^q$ has the unique 
critical point $z(q)$ on $h$, at which $\textup{Re}\,\varphi$ attains its minimal value. A direct computation shows that the Hessian $\textup{Hess}\,\varphi$ vanishes
on $V$ simultaneously with the Jacobian of the logarithmic Gauss mapping. Since $\partial E$ is simple, this Jacobian is not equal to zero at $z(q)$,
and hence $z(q)$ is a Morse critical point for the phase $\varphi$. Using the principle of stationary phase we obtain  formula (\ref{eq:asymp}) with
the constant $C(q)$ being the value at the point $z(q)$ of the function $P/z_1\cdot\ldots\cdot z_n\cdot Q'_{z_n}\cdot(\textup{Hess}\,\varphi)^{1/2}$. \qed
\end{proof}

\section{The thermodynamic ensemble and its most probable distribution}
\label{sec:4}

We consider a \emph{thermodynamic ensemble} $\mathfrak{U}$, consisting of $N$ copies of some physical system. Usually (see for instance
\cite{DF1}, \cite{DF2}, \cite{Fedoruk}, \cite{Sh} or \cite{Zo}) the system is characterized by energy values from a spectrum 	
$$
0=\varepsilon_0 < \varepsilon_1 < \varepsilon_2 <\ldots,\quad \varepsilon_j\in \mathbb{Z}\,.
$$
Each choice of energies in the systems of the ensemble defines a state of the ensemble. A basic question in the study of the behavior of an ensemble 
concerns the preferred states of the ensemble, as $N\to\infty$.

We will consider a more general situation, where the system is characterized by a multidimensional quantity $\varepsilon_k=(\varepsilon^1_k,\ldots, \varepsilon^n_k)$
from a given spectrum 
$$
\mathfrak{S}=\{\varepsilon_k\}_{k=\overline{0,\infty}}\subset \mathbb{N}^n\,,
$$ 
in which we for convenience shall assume that $\varepsilon_0=\bar{0}$. Futhermore, we shall consider
spectra from the lattice $\mathbb{Z}^n$, which lies in some acute cone in $\mathbb{R}^n\supset\mathbb{Z}^n$.

We introduce the quantity
\begin{equation}
\label{eq:1}
W(a)=W(a_0,a_1,\dots)=\frac{N!}{a_0!a_1!a_2!\ldots}\,,
\end{equation}
expressing the number of different states of the ensemble, for which exactly $a_k$ of the systems is in the state with parameter value
$\varepsilon_k$. We also say that $a_k$ is the $\varepsilon_k$ energy \emph{occupation number} in the ensemble. It is clear that in
(\ref{eq:1}) one should have
\begin{equation}\label{eq:2}
\sum_k a_k=N\,,
\end{equation}
\begin{equation}
\label{eq:3}
\sum_k a_k \varepsilon_k=\mathcal{E}\,,
\end{equation}
where $\mathcal{E}=(\mathcal{E}_1,\dots, \mathcal{E}_n)$ is the energy of the ensemble and the summation is over the index $k$ that
enumerates the elements $\varepsilon_k$ of the spectrum. The collection of numbers $a=(a_k)$ is said to be \emph{admissible} if it
satisfies conditions (\ref{eq:2}) and (\ref{eq:3}).

By definition, the \emph{most probable distributions} of energies among the systems of the ensemble (for $N\gg 1$) correspond to those $a$
that occur most frequently, that is, those that realize the maximum
$$
\max_a W(a)
$$ 
among all admisible collections $a$.

When considering the problem of describing the most probable energy distributions one makes the assumption that the vector $\mathcal{E}/N=u$ is
kept constant, that is, the average energy  $u=(u_1,\ldots, u_n)$ of the ensemble systems is fixed. Under this condition,  vector relation (\ref{eq:3})
written out coordinate-wise gives $n$ relations among the independent variables $a_k$. Just as in the case of a scalar spectrum ($n=1$, see for instance
\cite{Sh}), following the approach of Boltzmann, one uses the Langrange multiplier method to find the distributions that maximize $W(a)$, which we write now $W_u(a)$
(see \cite{PoTs} for details). The Lagrange multipliers $\mu_j$, that correspond to the coordinate-wise connections of  vector relation (\ref{eq:3}),
provide an important language for the solution of the assigned problem. More precisely, by introducing the \emph{partition function} as the series
$$
Z(\mu)=Z(\mu_1,\ldots, \mu_n)=\sum_{k} e^{-\left<\mu, \varepsilon_k\right>},
$$
we obtain the \emph{fundamental thermodynamic relations}:
$$
-\nabla_\mu \,\log Z=u\,,\quad a_k=N\frac{e^{-\left<\mu, \varepsilon_k\right>}}{Z},
$$
where $\nabla_\mu$ is the gradient with respect to the variables $\mu$.

In order to apply methods from analytic function theory and the method of stationary phase, it is more convenient for us to consider other (complex)
coordinates $z_j=e^{-\mu_j}$, $j=1,\ldots,n$. In these coordinates the partition function has the form
\begin{equation}\label{eq:5}
Z(z)=\sum_k z^{\varepsilon_k}=\sum_{\alpha\in \mathfrak{S}} z_1^{\alpha_1}\cdot \ldots \cdot z_n^{\alpha_n}.
\end{equation}
Analogously, the fundamental thermodynamic relations assume the form
\begin{equation}
\label{eq:Uj}
z_j\frac{Z'_{z_j}(z)}{Z(z)}=u_j\,,\quad j=1,\dots, n\,,
\end{equation}
\begin{equation}\label{eq:ak}
a_k=N\frac{z^{\varepsilon_k}}{Z(z)}\,.
\end{equation}
Let us give an interpretation  of these relations by the following
\begin{state}
\label{state:1}
\textit{For $N\gg 1$ the occupation values $a_k=a_k(u)$, computed from the formula (\ref{eq:ak}) in the solutions $z=z(u)$ of the system of equations (\ref{eq:Uj}), are the coordinates of the critical points for the function $W_u(a)$; in particular, the most probable 
distributions $a=(a_k)$ may be computed by means of the indicated formula for suitable solutions $z(u)$.}
\end{state}

The comparison of formulas (\ref{eq:Uj}) and (\ref{eq:n2}) shows that the solutions $z(u)$ to system~(\ref{eq:Uj})
is nothing more, than the inverse image $\gamma^{-1}(-u)$ of the logarithmic Gauss mapping $\gamma: \Gamma_Z\to \mathbb{CP}_n$
of the graph $\Gamma_f$ of the partition function $Z(z)$. However, the list of links between the mathematical notions
introduced in the first and the second sections and the fundamental thermodynamic relations goes beyond this shallow observation.
Another great of importance link can be figured out by the computation of critical values of the function $W_u(a)$.

Since the logarithm is a smooth function, the critical points of $W(a)$ and $\log W(a)$ coincide. The latter function can be
written for large $N$ with a help of Stirling's asymptotic formula in the form
$$
	\log W(a)= N\left(\log N - 1\right) -\sum_k a_k\left( \log a_k -1\right).
$$

The critical values of this function (under restriction $\mathcal{E}/N=u$) are
\begin{equation}
\label{eq:logWu}
	\log W_u=\log \left[ z(u)^{\mathcal{-E}} Z(z(u))^N\right]=N\left( \log Z(z(u)) - \left < u, \log z(u) \right > \right).
\end{equation}
It is easy to check this equality  by substitution in the previous expression for $\log W(a)$ 
of values (\ref{eq:ak}) for $a_k$ evaluated at the solutions $z=z(u)$ to the system~(\ref{eq:Uj}), using  relations
(\ref{eq:2}) and (\ref{eq:3}).

We are interested in the critical values $\log W_u$ only for real $u$, i.e. $u\in\mathbb{R}^n$.
The portion of a critical value attributed  to one system of an ensemble, i.e. the value
	$$
		S_u=: \frac{1}{N} \log W_u =\log Z(z(u))- \left< u, \log z(u) \right>
	$$
plays a role of  \textit{entropy}. Since in the logarithmic scale $\log z=-\mu$ one has
$$
	u=-\nabla_\mu \log Z =\nabla_{\log z} \log Z,
$$	  
 the entropy $S_u$ considered as a function in variables $u$ 
is the Legendre transform of the logarithm of a partition function in the logarithmic scale.

Thus, based on Proposition~\ref{prop:1} we get the following

\begin{state}
\label{state:2}
\textit{The liftings of the solutions $z(u)$ of the system (\ref{eq:Uj}), for 
$u\in\mathbb{R}^n \subset \mathbb{RP}_n$, to the graph $\Gamma_Z$ of the partition function coincide with inverse image $\gamma^{-1}(-u)$ of the Gauss logarithmic map
$\gamma: \Gamma_Z\to \mathbb{CP}_n$. On the amoeba $\mathcal{A}_{\Gamma_Z}$ of the graph these solutions parametrize the contour of the amoeba.  The values $S_u$ of the entropy coincide with the critical values of the linear function
$$
	l_u(x)=x_{n+1}-u_1 x_1 -\ldots -u_n x_n,
$$
restricted to the boundary $\partial E_{\bar{0},1}$ of the connected component $E_{\bar{0},1}$
of the complement $\mathbb{R}^{n+1}\setminus \mathcal{A}_{\Gamma_Z}$.}
\end{state}
	
For certain spectra $\mathfrak{S}$ the partition function $Z(z)$ admits an analytic continuation outside the domain of convergence of its series representation
(\ref{eq:5}) with new ``twin spectra'' $\mathfrak{S}'\subset \mathbb{Z}^n$ appearing.  Let us consider two examples.

\begin{example}
\label{exm:1}
The partition function $Z$ for the  spectrum
$\mathfrak{S}=\{0,2,3,4,\ldots\}$, $n=1$, is equal to the rational function $1+z^2/(1-z)$, which outside the unit disk $|z|<1$ has the development
$$
Z=-(z+\frac{1}{z}+\frac{1}{z^2}+\ldots)=-\sum_{\alpha\in\mathfrak{S}'}z^\alpha,
$$
where $\mathfrak{S}'=\{1,-1,-2,\ldots\}$. We can consider the thermodynamic relations (\ref{eq:Uj}),(\ref{eq:ak}) also in the complement  $\{|z|>1\}$ of the unit disk. The corresponding pieces of the amoeba of the graph of this rational function are depicted on Fig.~\ref{fig:3} in the middle.

\begin{figure}
\centering
\resizebox{0.75\textwidth}{!}{%
  \includegraphics{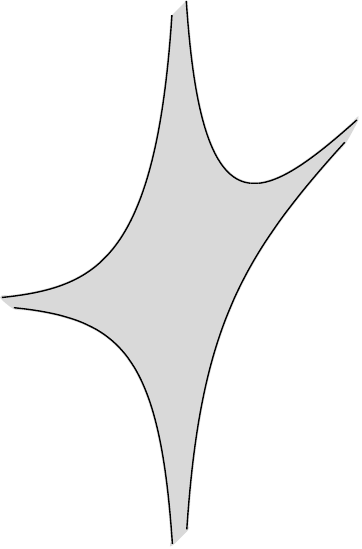}
  \includegraphics{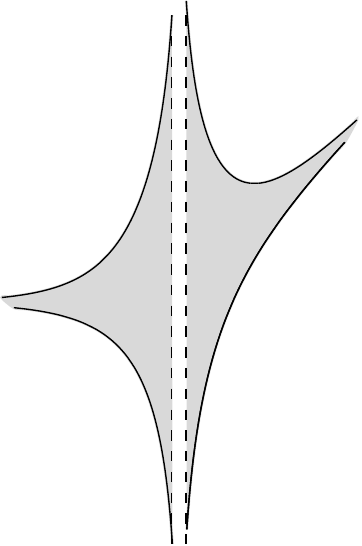}
  \includegraphics{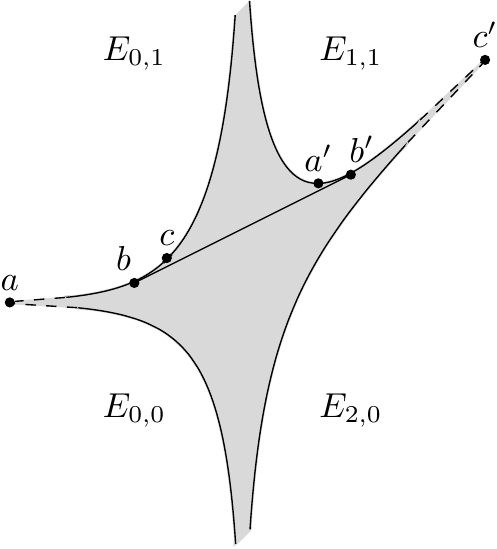}
}
\caption{Amoebas  for the graph of the partition function
$1+z^2/(1-z)$: the full amoeba (on the left), its pieces over  $|z|<1$ and over $|z|>1$ (in the middle)
and with a common tangent segment $[b,b']$ to $\partial E_{0,1}$ and $\partial E_{1,1}$ (on the right).}
\label{fig:3}       
\end{figure}

On Fig.~\ref{fig:3}, the points $a$ and $c'$ depict the points at infinity, where the normal 
vector $[-u:1]$  to the contour of the amoeba at $a$ and $c'$ equals $[0:1]$ 
and $[-1:1]$  respectively. 
The boundaries $\partial E_{0,1}$ and $\partial E_{1,1}$
have a common tangent at  the points $b$ and $b'$ (a simple computation shows that the normal vector $[-u_0:1]$ corresponds 
to the value  $u_0=1/2$). The set of the normal vectors to the arc $(a,b)\subset \partial E_{0,1}$
coincides with the set of the normal vectors to  $(a',b')\in \partial E_{1,1}$, the 
same holds for the pair of arcs $(b,c)$ and $(b',c')$. The tangents at the points of
the arc $(b,c)$  lie higher, than parallel to them tangents at the points of $(b',c')$,
the tangents at the points of $(a,b)$ and $(a', b')$. It follows
from Statement~\ref{state:2} that the maximal value of the entropy $S_u$ for $0<u<u_0$ corresponds to a
 solution $z(u)$ projected on the arc $(a',b')$, and it is from the domain $\{|z|>1\}$.

However, the combinatorial interpretation of $W(a)$ forbids us to consider the domain $\{|z|>1\}$,
because all the occupation  numbers in~(\ref{eq:ak}) for $z>1$ and some of them for $z<-1$ are negative.
Moreover, the partition function is negative at the points that project on the boundary $\partial E_{1,1}$.
\end{example}

The next example shows that in several dimensions we can overcome such  limitations.

\begin{example}
\label{exm:2}
Consider the spectrum
$$
	\mathfrak{S}=\{(0,0)\}\cup\{(2,2)+\mathbb{S}\}+\{(4,4)+\mathbb{S}\},
$$ 
where $\mathbb{S}$ is the semigroup $(2,1)\cdot\mathbb{N}+(1,2)\cdot\mathbb{N}$ (see Fig.~\ref{fig:4} on the left).

\begin{figure}
\centering
\resizebox{0.75\textwidth}{!}{%
  \includegraphics{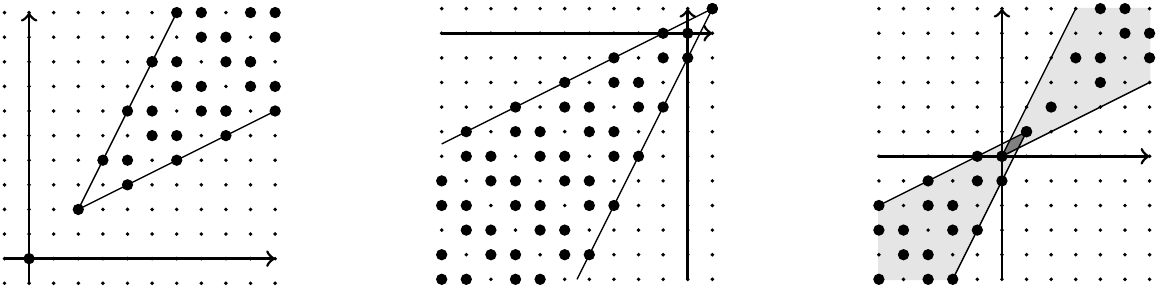}
 }
\caption{``Twin-spectra'' $\mathfrak{S}$ (on the left) and $\mathfrak{S}'$ (in the middle) and
their convex hulls (on the right).}
\label{fig:4}       
\end{figure}

The partition function $\sum_{\alpha\in\mathfrak{S}} z^\alpha$ converges in the domain $D=\{|z_1^2 z_2|<1$,
$|z_1 z_2^2|<1\}$ and equals
$$
	Z(z)=1+\frac{(1+z_1^2 z_2^2)z_1^2 z_2^2}{(1-z_1^2 z_2)(1-z_1 z_2^2)}.
$$ 
The development of $Z(z)$ in the domain $D'=\{|z_1^2 z_2|>1$, $|z_1 z_2^2|>1\}$ 
is again a partition function, i.e. it is a power series
$$
	Z(z)=\sum_{\alpha\in\mathfrak{S}'} z^\alpha
$$
with summation over the spectrum
$$\mathfrak{S}'=\{(0,0)\}\cup \{(-1,-1)-\mathbb{S}\}\cup \{(1,1)-\mathbb{S}\}$$ (see Fig.~\ref{fig:4} in the middle).

The full amoeba of the graph $\Gamma_Z$ corresponds to the polynomial
$$
	(w-1)(1-z_1^2z_2)(1-z_1z_2^2)-z_1^2z_2^2-z_1^4z_2^4
$$
in three variables $z_1, z_2, w$. Points $(0,0,1)$ and $(3,3,1)$ are vertices of the Newton polytope for this polynomial,
therefore the complement to the full amoeba of the graph $\Gamma_Z$ contains connected components $E_{0,0,1}$ and $E_{3,3,1}$.
Since the Laurent coefficients  of the developments of $Z$ in the domains $D$ and $D'$ are positive, the boundaries $\partial E_{0,0,1}$ and $\partial E_{3,3,1}$ are the $\Log$-images of the graph $\Gamma_Z$ over the real domains $D\cap \mathbb{R}^2_{+}$ and $D'\cap \mathbb{R}^2_{+}$. Consider the ``diagonal'' function
$$
	Z(t,t)=1+\frac{(1+t^4)t^4}{(1-t^3)^2}.
$$
The amoeba of its graph can be embedded in the amoeba $\mathcal{A}_{\Gamma_Z}$ by the mapping
$$
	i: (\log|t|, \log|Z(t,t)|)\mapsto (\log(t), \log(t), \log|Z(t,t)|).
$$
The boundaries of the components  $E_{0,1}$ and $E_{6,1}$ of the complement to the amoeba of the graph of $Z(t,t)$
are the $\Log$-images of pieces of the graph over the intervals $0<t<1$ and $1<t<\infty$, respectively.
The amoeba $\mathcal{A}_{\Gamma_Z}$  lives in the space $\mathbb{R}^3$ of variables $x_1, x_2, x_3$; and the plane
$x_1=x_2$ cuts out in the surfaces $\partial E_{0,0,1}$ and $\partial E_{3,3,1}$ two pieces, the images
$i(\partial E_{0,1})$ and $i(\partial E_{6,1})$,  respectively.
As in Example~\ref{exm:1}, the curves $\partial E_{0,1}$ and $\partial E_{6,1}$ have a common tangent line, lying below these curves,
since they are convex.

In view of the symmetry of $\mathcal{A}_{\Gamma_Z}$ with respect to the plane $x_1=x_2$, there
exists a common tangent plane $\tau$ to surfaces $\partial E_{0,0,1}$ and $\partial E_{3,3,1}$
with the property that $\tau$ crosses the common tangent line to the embeddings $i(\partial E_{0,1})$ and $i(\partial E_{6,1})$ symmetrically  with respect to the plane $x_1=x_2$. As it follows from results of Sect.~\ref{sec:6}, the vector $[u_1:u_2:1]$ is normal to the tangent plane $\tau$, 
if $u=(u_1, u_2)$ belongs to the intersection of interiors of convex hulls of the spectra $\mathfrak{S}$ and $\mathfrak{S}'$,
i.e. to the double-shaded rhombus on the right of Fig.~\ref{fig:4}. In general, the rhombus is divided by some curve $\gamma$ into two domains, such that the value of the entropy $S_u$ (corresponding to the ensemble with the spectrum $\mathfrak{S}$) 
is greater than that of the entropy  $S'_u$
(corresponding to the ensemble with the spectrum $\mathfrak{S}'$) in the first domain and is less in the second one. Perhaps,
this phenomenon may be considered as a tunnelling transition from one ensemble to another in a way to increase the entropy,
when we choose  the value of the energy  $u$ on $\gamma$. 
\end{example}

At the end of this section, we show that the notion of multidimensional spectrum, our starting point,
leads to the notions of the multidimensional temperature and the vector of thermodynamic forms. For this purpose,
we compute the differential of logarithm of a partition function assuming that the variables
$z_1,\ldots, z_n$ are positive and entries $\varepsilon_k$  of the spectrum $\{\varepsilon_k\}$
vary in some neighbourhood of lattice points in $\mathbb{R}^n$, i.e. we consider
the spectrum $\{\varepsilon_k\}$ to be variable.

In accordance with~(\ref{eq:Uj}) and (\ref{eq:ak})
$$
 d \log Z=(d_z+d_\varepsilon) \log Z=\sum_j z_j \frac{Z'_{z_j}}{Z}\frac{dz_j}{z_j}+\sum_k \sum_j \frac{Z'_{\varepsilon_k^j}}{Z} d\varepsilon_k^j=
$$
$$
	=\left< u, d\log z\right>+\sum_j \sum_k \frac{z^{\varepsilon_k^j}}{Z} \log z_j d\varepsilon_k^j=
	\left< u, d\log z\right> +\left< \log z, \frac{1}{N}\sum_k a_k d\varepsilon_k\right>.
$$
Hence, we get the following expression for the differential of the entropy 
$$
	d S=d \left( \log Z -\left<u,\log z\right>\right)=\left< -\log z, d u\right>
	+\left< \log z, \frac{1}{N}\sum_k a_k d\varepsilon_k\right>=\left<\frac{1}{T}, \omega\right>,
$$
where
$$
	\omega=(\omega_1,\ldots, \omega_n),\ T=(T_1,\ldots, T_n)
$$
denote the vector of the thermodynamic forms and the vector of the temperature with components
$$
	\omega_j=d u_j - 1/N \sum_k a_k d \varepsilon_k,\ T_j=-1/\log z_j.
$$

\section{The average value $\overline{a}$ of the admissible collections $\{a\}$}
\label{sec:5}

In the preceding section we gave a description, following Botzmann, of the most probable distributions of the ensemble. However, the method that was used is 
somewhat limited, since the extremal points (\ref{eq:ak}) for (\ref{eq:1}) are obtained by applying the Stirling formula to $a_k!$, and this is only justified for large
values of $a_k$. In the case of a scalar spectrum, the Darwin--Fowler method offers a possibility to avoid this drawback. It consists in a description of the
asymptotics of the averages of the occupation numbers. We shall analogously describe the asymptotics of the averages of the occupation numbers, when the
energy spectrum is composed of vector quantities. In this section we show that this problem is equivalent to the problem of describing the asymptotics of the
diagonal coefficients of a Laurent expansion of the meromorphic function $w/(w-Z(z))$.

\begin{definition}[\cite{DF1}, \cite{Sh}]
The \emph{average value of the admissible collections} $\{a\}$ is the collection $\overline{a}=(\overline{a}_k)$ of numbers
$$ 	\overline{a}_k=\frac{\sum_{a} a_k W(a)}{\sum_{a} W(a)}\,,$$
where the summation is over all admissible collections $a=(a_k)$. 
\end{definition}

For the study of the averages $\overline{a}_k$ we introduce the sum
\begin{equation}
\label{eq:7}
\sum_{a} W(a,\omega)=\sum_{a}\frac{N!}{a_0!a_1!\dots a_k!\dots} \omega_0^{a_0}\omega_1^{a_1}\dots \omega_k^{a_k}\dots
\end{equation}
over all admissible collections $a=(a_k)$. Here the $\omega_j$ are real parameters, all varying in a small neighborhood of $1$.
We remark that $W(a,I)=W(a)$, where $I=(1,1,\dots)$ is the all ones vector. Hence, for $\omega=I$ the quantity (\ref{eq:7}) expresses
the \emph{total number of states} of the ensemble. It is not difficult to see that
\begin{equation}\label{eq:new_star}
\overline{a}_k=\left.\frac{\partial}{\partial\omega_k}\log\sum_a W(a,\omega)\right|_{\omega=I}.
\end{equation}

As in \cite{DF1} and \cite{Sh} one proves the integral representation
\begin{equation}\label{eq:8}
\sum_{a} W(a,\omega)=\frac{1}{(2\pi \imath)^n} \int_{T_r} f^N(z) z^{-\mathcal{E}}\bigwedge_1^n\frac{dz_j}{z_j}\,,	
\end{equation}
where $T_r=\{|z_1|=r_1,\dots, |z_n|=r_n\},$ and the $r_j$ are chosen so small that on $T_r$ one has convergence of the series
$$
f(z)=f(z,\omega)=\sum_{k} \omega_k z^{\varepsilon_k}=\sum_{k} 
\omega_k z_1^{\varepsilon_k^1}\cdot\ldots\cdot z_n^{\varepsilon_k^n}\,.
$$
Since $f(z,I)=Z(z)$ we refer this series to be a variation of partition function.
Since the condition $0<\omega_k< 1+\delta$ is fulfilled, the domain of convergence $G'$ of this series is non-empty and contains 
the origin $z=0$.

We now introduce the function of $n+1$ variables
$$F(z,w)=\frac{w}{w-f(z)}\,,$$ 
which is meromorphic in the domain $G=G'\times\mathbb{C}_w$.  The polar hypersurface of $F$ is the graph
$$
\Gamma_f=\{(z,w)\in G: w=f(z)\}\,.
$$
Due to the fact that $\varepsilon_0=\bar{0}$, the closure $\mathcal{N}$ of the convex hull of the summation support of the
series $w-f(z)$ contains the vertex $\nu=(\bar{0},1)$. According to Theorem~\ref{thm:1} this vertex corresponds to a connected component $E_{\bar{0},1}$ of the complement of the amoeba $\mathcal{A}_V$. Using a geometric progression we expand $F$ in a Laurent series, convergent in $\Log^{-1} (E_{\bar{0},1}) \subset \{(z,w)\in G: |w|>|f(z)|\}$:

\begin{equation} \label{eq:La}
F(z,w)=\sum_{N=0}^{\infty}\frac{f^N}{w^N}=\sum_N \sum_\mathcal{E} C_{\mathcal{E},-N} z^\mathcal{E} w^{-N}.
\end{equation}

For the Laurent coefficients $C_{\mathcal{E},-N}$ of this series one has the integral representation
$$
C_{\mathcal{E},-N}=\frac{1}{(2\pi\imath)^{n+1}} \int_{\Log^{-1}(x)}\frac{w}{w-f(z)} z^{-\mathcal{E}}w^N \bigwedge_1^n \frac{dz_j}{z_j}\wedge\frac{dw}{w}\,,
$$
where $x\in E_{\bar{0},1}$. Performing the integration with respect to $w$ in this last integral, we immediately obtain (\ref{eq:8}).

We thus find that the problem of describing the asymptotics of the sum (\ref{eq:7}) is equivalent to the same problem for the coefficients $C_{\mathcal{E},-N}$
of the series (\ref{eq:La}), for $\mathcal{E}=u\cdot N$, with $u$ being the vector of average energies. That is, it amounts to finding the asymptotics of the
diagonal coefficients $C_{(u,-1)\cdot N}$ with direction vector $q=(u,-1)$.

\section{The asymptotics of the average values $\overline{a}_k$}
\label{sec:6}

Let the point $(z_*,w_*)$   on the graph $\Gamma_f$ of the variation $f$ of partition function be such that $\Log (z_*,w_*)\in\partial E_{\bar{0},1}$. 

Since $\partial E_{0,1}$ is a part of the amoeba contour, the first coordinates $z_*$ of the given point on the graph satisfy (\ref{eq:n2}) for some $u\in\Real^n_{+}$, and the coordinate
$w_*$ is uniquely determined by $z_*$. As we let $\omega$ tend to the vector $I=(1,1,\ldots)$, we get $f\to Z$, and the point
$(z_*,w_*)$ moves to the point $(z,w)=(z(u),w(u))$, whose logarithmic image lies on the boundary $\partial E_{\bar{0},1}$ of the 
component $E_{\bar{0},1}$ of the complement to the amoeba of the graph $\Gamma_Z=\{w=Z(z)\}$ of the partition function of the ensemble. Besides that, $z(u)$ satisfies system~(\ref{eq:Uj}).

\begin{theorem}
\label{thm:3}
Suppose that the spectrum $\mathfrak{S}=\{\varepsilon_k\}$ generates the lattice $\mathbb{Z}^n$, and that the point
$z=z(u)\in\Real^n_{+}$ satisfies the system (\ref{eq:Uj}). Then, as $N\to\infty$, the average values $\overline{a}_k$ for the occupation numbers
of energy $\varepsilon_k$ has the form
\begin{equation}\label{eq:new_starstar}
\overline{a}_k\sim N \left.\frac{z^{\varepsilon_k}}{Z(z)}\right|_{z=z(u)}
\end{equation}
and they coincide with most probable values of $a_k$.
\end{theorem}

\begin{proof}

By assumption the spectrum $\mathfrak{S}$ generates the lattice $\mathbb{Z}^n$ and hence, according to Proposition~\ref{prop:3}
the boundary $\partial E_{\bar{0},1}$ is simple. Therefore we can apply Theorem~\ref{thm:2} for the asymptotics of the
diagonal sequence of Laurent coefficients of the series (\ref{eq:La}):
$$
C_{(u,-1)\cdot N}\sim C(q)\cdot N^{-\frac{n}{2}}\cdot (z_*^{-u}(u)w_*(u))^{N},\quad N\to+\infty\,.
$$

Hence, taking into account the summary in Sect.~\ref{sec:5}, we find that the asymptotics of the total number of states, as $N\to+\infty$,
has the form
$$
\sum_a W(a,\omega)\sim C(q)\cdot N^{-\frac{n}{2}}\cdot (z_*^{-u}(u)\cdot f(z_*(u)))^{N}.
$$

Now, direct calculation leads us to the asymptotic equality
$$
\frac{\partial}{\partial \omega_k} \log \sum_a W(a,\omega)\ \sim \
N\cdot \left<\nabla_z \varphi(z_*(u)),\frac{\partial}{\partial\omega_k }z_*(u)\right>+
N\cdot\frac{z^{\varepsilon_k}_*(u)}{f(z_*(u))}\,,
$$
where $\varphi=\log(z^{-u}f(z))$ denotes the phase (see the proof of Theorem~\ref{thm:2}). In the right hand side of the last formula
the first term is equal to zero, because $z_*$ is a critical point for the phase $\varphi$. Therefore, setting $\omega=I$, we get from 
the formula (\ref{eq:new_star}) the desired asymptotics (\ref{eq:new_starstar}).\qed
\end{proof}

Let us now raise the question about what the admissible values are for the vector $u$ of average energies, that guarantee the existence of a solution $z(u)\in \Real_{+}^n$ to the system (\ref{eq:Uj}), and hence provide the asymptotics (\ref{eq:new_starstar}).

In the work of Darwin and Fowler \cite{DF1},\cite{DF2} this question was not considered. Apparently, it was first paid attention to in
\cite[Sect.~4.5.1]{Fedoruk}, where it was observed that if the partition function is a polynomial of degree $d$, then the admissible
average energies must be taken within the interval $0<u<d$, that is, in the interior of the convex hull of the numbers
$0=\varepsilon_0<\varepsilon_1<\ldots<\varepsilon_k=d$.

The raised question is answered by the following theorem, where we use the notation $\mathcal{N}^{\circ}$ for the interior of the
convex hull in $\Real^n$ of the spectrum $\mathfrak{S}=\{\varepsilon_k\}$.
\begin{theorem} 
\label{thm:4}
Suppose that the spectrum $\mathfrak{S}=\{\varepsilon_k\}$ generates the lattice $\mathbb{Z}^n$. Then for every value of the average 
energy $u\in \mathcal{N}^{\circ}$ the system (\ref{eq:Uj}) has a unique solution $z=z(u)$ in $\mathbb{R}^n_{+}$, and hence for
$u\in \mathcal{N}^{\circ}$ the average values $\overline{a_k}$ coincide with the most probable ones. 
\end{theorem}

\begin{proof}
Lifting the solutions $z(u)$ of the system of equations (\ref{eq:Uj}) for $u\in\Real^n$ to the graph $\Gamma_Z$ of the partition 
function of the ensemble, we obtain the criticial values for the mapping $\left. \Log\right|_V$. On the amoeba $\mathcal{A}_{\Gamma_Z}$ of the graph, these solutions parametrize its contour. In particular, the solutions $z(u)\in\Real^n_{+}$ that are of interest to us parametrize
the boundary of the complement component $E_{\bar{0},1}$. Thanks to the fact that the spectrum generates the lattice $\mathbb{Z}^n$,
we know from Proposition~\ref{prop:3} that to each point on $\partial E_{\bar{0},1}$ there corresponds a unique vector $q\in K_{E_{\bar{0},1}}$. 
Therefore, in order to obtain all solutions $z(u)$ from $\Real^n_{+}$ one must go through all vectors $q$ from the component cone $E_{\bar{0},1}$.

\begin{figure}
\centering
\resizebox{0.75\textwidth}{!}{%
  \includegraphics{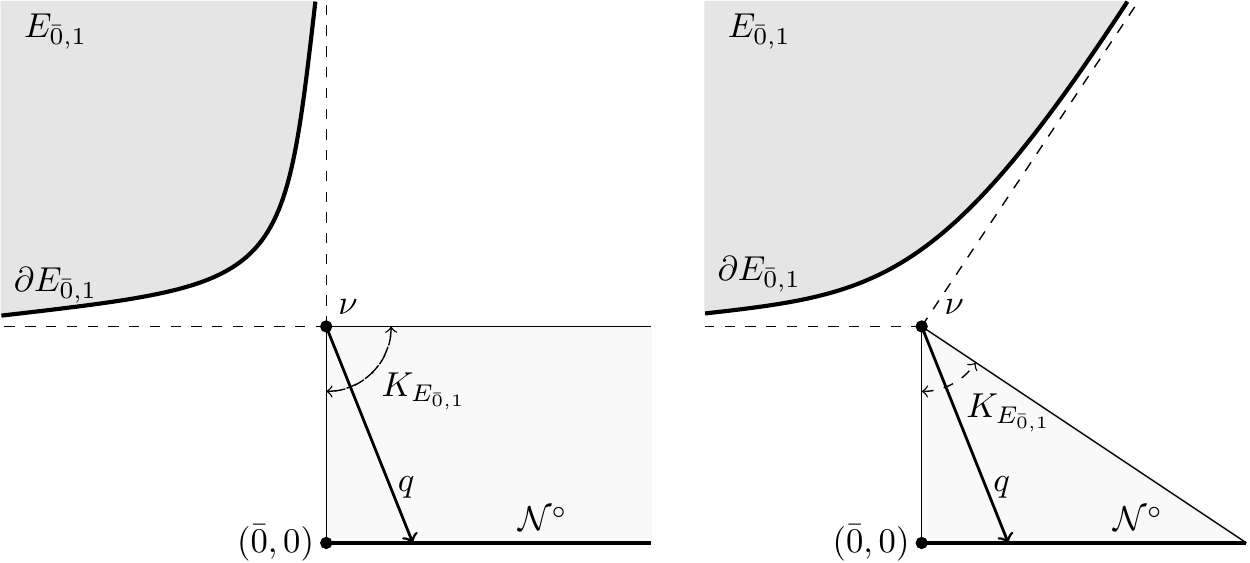}
}
\caption{The relations between $\mathcal{N^\circ}, K_{E_{\bar{0},1}}$ and $E_{\bar{0},1}$ for a finite (left) 
and an infinite (right) spectrum.}
\label{fig:5}       
\end{figure}

By Theorem~\ref{thm:1} the recession cone of the component $E_{\bar{0},1}$ is the dual cone to $\hat{\mathcal{N}}$ at the vertex $\nu=(\bar{0},1)$,
where $\hat{\mathcal{N}}$ denotes the closure of the convex hull of the summation support of the series $Q=w-Z(z)$. (See Figure~\ref{fig:5} where the
recession cone is bounded by dashed lines.) The outward normals of those facets of the polyhedron $\hat{\mathcal{N}}$ that come together at the vertex $\nu$ span this dual cone. Therefore, the sought cone $K_{E_{\bar{0},1}}$ is spanned by the edges of $\hat{\mathcal{N}}$ that emanate from the vertex $\nu$, and thus $K_{E_{\bar{0},1}}$ consists of all vectors of the form $q=(u,-1)$,
with $u\in\mathcal{N}^\circ$.\qed
\end{proof}
 
We conclude with some remarks and illustrations to Theorem~\ref{thm:4}.
First, the statement of the theorem still holds if one shifts the spectrum $\mathfrak{S}$
by a noninteger vector. For example, the domain of admissible average values of energy in the case
of the Plank oscillator with the spectrum $\{1/2+\mathbb{N}\}$ equals $\{u>1/2\}$.  Such domain for the Fermi
oscillator with the spectrum $\mathfrak{S}=\{0,1\}$ is the interval $\{0<u<1\}$ (see \cite[ch.~4]{Sh}).
The latter case is depicted on the right of Fig.~\ref{fig:5}. Example~\ref{exm:2} of Sect.~\ref{sec:4} deals
with the ``twin-spectra'', and the sectors on Fig.~\ref{fig:4} are the domains of admissible average values of energy
in the corresponding cases. These sectors have a nonempty intersection, the double-shaded
rhombus (Fig.~\ref{fig:4}, on the right).
\begin{acknowledgements} 

The second author was supported by RFBR  grant 09-09-00762 and ``M\"obius Competition'' fund for support to  young scientists. The third author was  supported by the Russian Presidential grant N\v S-7347.2010.1 and by RFBR 11-01-00852.
\end{acknowledgements}

\end{document}